\documentclass[11pt,letterpaper]{article}
\usepackage{t1enc}
\usepackage{mathptmx}
\usepackage{enumitem}
\usepackage[leqno]{amsmath}
\usepackage{bm}
\usepackage{amsfonts}
\usepackage{comment}
\usepackage{url}
\usepackage{amssymb,amsthm}
\usepackage{latexsym}
\usepackage{geometry}
\usepackage[numbers,sort&compress]{natbib}
\usepackage{color}
\usepackage{graphicx}
\numberwithin{equation}{section}

\newcommand{\om}{\omega}
\newcommand{\z}{\emptyset}
\renewcommand{\phi}{\varphi}
\newcommand{\Po}{\mathbf{Po}}			
\newcommand{\ar}{``$\Implies$'': \ }
\newcommand{\al}{``$\Leftarrow$'': \ }
\newcommand{\Implies}{\Rightarrow}
\newcommand{\Iff}{\Longleftrightarrow}
\newcommand{\da}{\downarrow}

\newcommand{\R}{\ensuremath{\mathcal{R}}}                       
\newcommand{\U}{\ensuremath{\mathcal{U}}}
\newcommand{\cO}{\ensuremath{\mathcal{O}}}

\newcommand{\Alp}{\ensuremath{\mathcal{A}}}
\newcommand{\A}{\ensuremath{\mathfrak{A}}}
\newcommand{\fA}{\ensuremath{\mathfrak{A}}}
\newcommand{\B}{\ensuremath{\mathfrak{B}}}
\newcommand{\C}{\ensuremath{\mathfrak{C}}}
\newcommand{\D}{\ensuremath{\mathfrak{D}}}

\newcommand{\N}{\ensuremath{\mathbb{N}}}
\newcommand{\V}{\ensuremath{\mathbf{V}}}

\newcommand{\cL}{\ensuremath{\mathcal{L}}}
\newcommand{\K}{\operatorname{{K}}}
\newcommand{\fin}{\mathsf{fin}}
\newcommand{\inv}[1]{#1^{-1}}

\newcommand{\Cm}{\ensuremath{\operatorname{{\mathfrak{Cm}}}}}   
\newcommand{\CmN}{\Cm~\N}
\newcommand{\CmNN}{\Cm_0\N}
\newcommand{\CmEN}{\Cm_1\N}
\newcommand{\At}{\operatorname{\mathfrak{At}}}

\newcommand{\BH}{\mbox{\sc BH}}
\newcommand{\cE}{\ensuremath{\mathcal{E}}}       

\newcommand{\df}{\ensuremath{:=}}
\newcommand{\klam}[1]{\ensuremath{\langle #1 \rangle}}
\newcommand{\set}[1]{\ensuremath{\{#1\}}}
\newcommand{\card}[1]{\ensuremath{\lvert #1 \rvert}}
\newcommand{\cmp}[1]{\ensuremath{\overline{#1}}}
\newcommand{\symdiff}{\vartriangle}
\newcommand{\onto}{\twoheadrightarrow}
\newcommand{\po}{\ensuremath{2^\omega}}

\newcommand{\In}{\operatorname{In}}

\newcommand{\tand}{\text{ and }}

\newcommand{\tiff}{if and only if \ }
\newcommand{\wlg}{w.l.o.g.~}

\renewcommand{\hom}{\mathbf{H}}
\newcommand{\sub}{\mathbf{S}}
\newcommand{\iso}{\mathbf{I}}
\newcommand{\prods}{\mathbf{P}}

\newcommand{\Var}{\ensuremath{\mathbf{Var}}}   
\newcommand{\Eq}{\mathbf{Eq}~}                  
\newcommand{\Eqs}{\ensuremath{\mathcal{E}}}
\newcommand{\EqSat}{\mathbf{EqSat}~}                  
\newcommand{\FO}{\mathbf{FO}~}

\newcommand{\cPlus}{\ensuremath{\bm{+}}}         
\newcommand{\cTimes}{\ensuremath{\bullet}}       

\newcommand{\plus}{\ensuremath{\bm{+}}}         
\newcommand{\mal}{\ensuremath{\bullet}}         

\theoremstyle{plain}
\newtheorem{thm}{Theorem}[section]
\newtheorem{lemma}[thm]{Lemma}
\newtheorem{cor}[thm]{Corollary}

\theoremstyle{definition}

\newtheorem*{claim*}{Claim}
\newtheorem{question}{Question}

\date{}

\title{Complex algebras of arithmetic%
\thanks{The authors gratefully acknowledge the support of the EPSRC, grant
number EP/F069154/1. Ivo D{\"u}ntsch also acknowledges support from the
Natural Sciences and Engineering Research Council of Canada.}}
\author{Ivo D\"untsch \\
Department of Computer Science,\\
Brock University,\\
St.\ Catharines, ON, L2S 3A1\\
Canada \\
\url{duentsch@brocku.ca} \and
Ian Pratt--Hartmann \\
School of Computer Science, \\
University of Manchester,\\
Oxford Road, Manchester, M13 9PL \\
United Kingdom, \\
\url{ipratt@cs.man.ac.uk}
}

\parskip=\bigskipamount
\begin{document}
\maketitle
\thispagestyle{empty}

\begin{abstract}
\noindent
An {\em arithmetic circuit} is a labeled, acyclic directed graph specifying a sequence of arithmetic and logical operations to be performed on sets of natural numbers. Arithmetic circuits can also be viewed as the elements of the smallest subalgebra of the complex algebra of the semiring of natural numbers. In the present paper we investigate the algebraic structure of complex algebras of natural numbers and make some observations regarding the complexity of various theories of such algebras.
\end{abstract}

\section{Introduction}

Let $\omega$ be the set of natural numbers $\set{0,1,2, \ldots}$. An {\em arithmetic circuit} (AC) \cite{McKenzieW03,McKenzieW07} is a labeled, acyclic directed graph specifying a sequence of arithmetic and logical operations to be performed on sets of natural numbers. Each node in this graph evaluates to a set of natural numbers, representing a stage of the computation performed by the circuit. Nodes without predecessors in the graph are called {\em input nodes}, and their labels are singleton sets of natural numbers. Nodes with predecessors in the graph are called {\em arithmetic gates}, and their labels indicate operations to be performed on the values of their immediate predecessors; the results of these operations are then taken to be the values of the arithmetic gates in question. One of the nodes in the graph (usually, a node with no successors) is designated as the {\em circuit output}; the set of natural numbers to which it evaluates is taken to be the value of the circuit as a whole.

More formally, an arithmetic circuit is a structure $C = \klam{G,E,g_C, \alpha}$,
where $\klam{G,E}$ is a finite acyclic and asymmetric graph over $2^\omega$, $\In(g) \leq 2$ for all $g \in G$,
and $\alpha: G \to \set{\cup, \cap, {}^-, \plus, \mal} \cup \set{\set{n}: n \in \omega} \cup \set{\z, \omega}$ is a labeling function for which
\begin{gather}
\alpha(g) \in
\begin{cases}
\set{\set{n}: n \in \omega} \cup \set{\z, \N}, &\text{if } \In(g) = 0, \\
\set{{}^-}, &\text{if } \In(g) = 1, \\
\set{\cup, \cap, \plus, \mal}, &\text{if } \In(g) = 2.
\end{cases}
\end{gather}
Here, $\In(g)$ is the in--degree of $g$ and $\plus$ and $\mal$ are the complex extensions of $+$ and $\cdot$, i.e.
\begin{gather}
a \plus b \df \set{k +n: k \in a, \ n \in b}, \ a \mal b \df \set{k \cdot n: k \in a, \ n \in b}.
\end{gather}
$g_C$  is called the \emph{output gate}; if $\In(g) = 0$, we call $g$ an \emph{input gate} or a \emph{source}. 

The arithmetical interpretation of $C$ is as follows:
\begin{enumerate}[label=\textup{(}\emph{\roman*}\textup{)}]
\item If $\In(g) = 0$, then $I(g) = \alpha(g)$.
\item If $\In(g) = 1$, and $g'$ is the unique predecessor of $g$, then $I(g) =
\N \setminus {I(g')}$.
\item If $\In(g) = 2$, and $g_0, g_1$ are the two predecessors of $g$, then $I(g) \df I(g_0)~\alpha(g)~I(g_1)$.
\end{enumerate}
 $I(C)$ is defined as $I(g_C)$.

Fig.~\ref{fig:evensPrimes} shows two examples of arithmetic circuits, where the output gate is indicated by the double circle. In Fig.~\ref{fig:evensPrimes}a, Node~1 evaluates to $\set{1}$, and Node~2 to $\omega$; hence, Node~3 evaluates to $\set{1} \cPlus \set{1} = \set{2}$, and Node~4, the output of the circuit, to $\set{2} \cTimes \omega$, i.e.~the set of even numbers. The circuit of Fig.~\ref{fig:evensPrimes}b functions similarly: Node~2 evaluates to $\set{0} \cup \set{n \in \omega: n \geq 2}$, and Node~3 to $\set{0} \cup \set{n \in \omega: \text{$n$ is composite}}$; hence, Node~4 evaluates to the set numbers which are either prime or equal to 1, and Node~5, the output of the circuit, to the set of primes. We say that the circuits of Fig.~\ref{fig:evensPrimes}a and Fig.~\ref{fig:evensPrimes}b {\em define}, respectively, the set of even numbers and the set of primes. Any arithmetic circuit defines a set of numbers in this way.

\begin{figure}[htb]
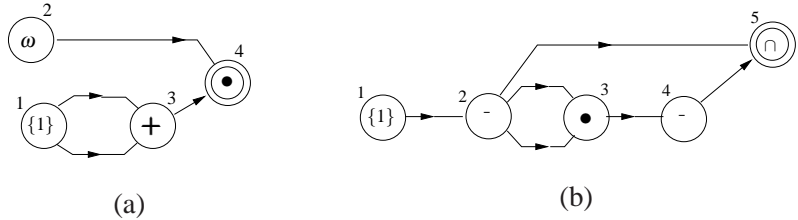

\label{fig:evensPrimes}
\caption{Arithmetic circuits defining: (a) the set of even numbers; (b)
  the set of primes. The integers next to the nodes are for reference only.}
\begin{center}

\vspace{0.15cm}

\begin{minipage}{3.5cm}
\begin{center}
\input{evens.pstex_t}

\vspace{0.25cm}

(a)
\end{center}
\end{minipage}
\hspace{1cm}
\begin{minipage}{6cm}
\begin{center}
\input{primes.pstex_t}

\vspace{0.25cm}

(b)
\end{center}
\end{minipage}

\vspace{-0.5cm}

\end{center}
\end{figure}

If $\cO \subseteq \set{\cap,\cup,{}^-,\plus,\mal}$, an \emph{$\cO$ -- circuit} is an arithmetic circuit whose non--input labels are among those contained in $\cO$. Let
\begin{gather}\label{memprob}
MC(\cO) = \set{\klam{C,n}: C \text{ is an $\cO$ -- circuit}, n \in I(C)}.
\end{gather}
The \emph{membership problem for $\cO$} is the question whether $MC(\cO)$ is decidable \cite{McKenzieW03}. In other words, is there an algorithm which decides membership of an arbitrary $n \in \omega$ in an arbitrary output $C$ of an $\cO$ -- circuit? If the problem is decidable, then its complexity is of interest. For almost all cases of $\cO$, the complexities have been determined by \citet{McKenzieW03}. The question whether $MC(\cO)$ is decidable where $\cO = \set{\cap,\cup,{}^-,\plus,\mal}$ is still open. The table of complexities for the membership problem where all Boolean operators are present is given in Table \ref{fig:comp1}.

\begin{table}[htb]
\caption{Complexity results for MC \cite{McKenzieW03}}\label{fig:comp1}
$$
\begin{array}{|l|c|c|} \hline
\cO & \text{Lower bound} & \text{Upper bound} \\ \hline
\cap,\cup,{}^-,\plus & \text{PSPACE} & \text{PSPACE} \\
\cap,\cup,{}^-,\mal & \text{PSPACE} & \text{PSPACE} \\
\cap,\cup,{}^-,\plus,\mal & \text{NEXPTIME} & ? \\ \hline
\end{array}
$$
\end{table}

Algebraically speaking, an arithmetic circuit can be regarded as a well -- formed term over an alphabet $\Alp$ containing operations from $\set{\cap,\cup,{}^-,\z, \omega,\plus,\mal}$ and constants from $\set{\set{n}: n \in \omega}$ as input gates. If $\plus$ is present, then $\set{0}$ will suffice since
\begin{gather}
\set{1} = \cmp{\cmp{\set{0}} \plus \cmp{\set{0}}} \cap \cmp{\set{0}}.
\end{gather}

The membership problem now can be seen as a word problem over $\Alp$:
\begin{align}
& \text{Given $n \in \omega$ and a well formed term $\tau$ over $\Alp$, is $\set{n} \cap \tau = \set{n}$}? \label{mp1}
\end{align}
It is natural to generalize the notion of arithmetic circuits by allowing input nodes to represent {\em variable}
sets of numbers \cite{GRTW07}. Logically speaking, we enhance our language by a set $V$ of variables which are interpreted as sets of natural numbers; arithmetic circuits correspond to the variable free terms of this language.
It now makes sense to consider satisfiability and validity of (in--) equations of terms of this language under this interpretation. Furthermore, the operations $f: (\po)^k \rightarrow \po$ definable from the given operators $\cO$ can be studied \cite{pd_acfunc}.

In analogy to the membership problem, \citet{GRTW07} consider the complexity of
\begin{xalignat*}{2}
SC(\cO) &= \set{\klam{C(x_0, \ldots,x_n),k}: C \text{ is an $\cO$ circuit and }
(\exists k_0, \ldots,k_n)[k \in I(C({k_0}, \ldots,{k_n})]}&& 
\end{xalignat*}
for various sets $\cO$ and determine many of these complexities. The main open problem is the question whether $SC(\cap,\cup,{}^-, \mal)$ is decidable. In other words, is it decidable whether the equation

\begin{gather}\label{sc}
\set{k} \cap \tau(x_0,\ldots,x_{n-1}) = \set{k}
\end{gather}
has a solution over the subsets of $\om^n$?

In this paper we shall shed some light on these question and the structure of arithmetic circuits from an algebraic viewpoint. Our main tool will be the apparatus of Boolean algebras with operators,
in particular, complex algebras of first order structures, which were introduced by \citet{jt51}.

\section{Notation and definitions}
\subsection{Algebras}

An algebra $\A$ is a pair $\A = \klam{A,\cO}$, where $A$ is a set and $\cO = \set{f_i: i \in I}$ a set of operation symbols
$f$ each having a finite arity $\alpha(f)$; if we write $f(x_0, \ldots, x_{n-1})$ we implicitly assume that $\alpha(f) = n$.
Operations of arity $0$ are called \textit{(individual) constants}. We will usually denote
algebras by gothic letters $\A, \B, \ldots$, and their universes by the corresponding roman
letter $A,B, \ldots$. $\A$ is called \emph{subdirectly irreducible} if it has a smallest nontrivial congruence, and \emph{congruence--distributive} if its congruence lattice is distributive.

Suppose that $\K$ is a class of algebras (of the same type $\cO$). For $\A,\B \in K$, $\A \leq \B$ means that $\A$ is a subalgebra of $\B$. The operators $\iso, \sub, \hom$ and $\prods$ have their usual meaning. $\Var(\K)$ is the variety generated by $\K$, i.e. $\Var(\K) = \hom\sub\prods(\K)$. A variety $\V$ is called \emph{finitely based} if there is a finite set $\Sigma$ of equations in the language of $\V$ such that $\A \in \V$ \tiff $\A \models \Sigma$, and $\V$ is called \emph{finitely generated} if there is a finite set $\K$ of finite algebras such that $\V = \Var(\K)$.

Suppose that $\K$ is a class of algebras of the same type $\cO$. We consider the following sets of formulas in the language of ${\cO}$ (plus equality).
\begin{enumerate}[label=\textup{(}\emph{\roman*}\textup{)}]
\item The {\em first-order theory $\FO\K$ of $\K$}: The set of first-order formulas true in each member of $\K$.
\item The {\em equational theory $\Eq\K$ of} $\K$: The set of formulas of the forms $\tau(x_0, \ldots,x_n) = \sigma(x_0, \ldots,x_n)$ whose universal closures are true in each member of $\K$.
\item The {\em satisfiable equations $\EqSat\K$ of} $\K$: The set of formulas of the forms $\tau(x_0, \ldots,x_n) = \sigma(x_0, \ldots,x_n)$ whose existential closures are true in each member of $\K$.
\end{enumerate}
If $\K = \set{\A}$, we usually write $\FO\A$, $\Eq\A$, etc.

\subsection{Boolean algebras with operators}

In the following, let $\B = \klam{B,{\lor, \land, {}^-, \bot, \top}}$ be a Boolean algebra (BA); here, $\bot$ is the smallest and $\top$ is the largest element of $B$. If $a,b \in B$, then $a \symdiff b$ denotes the symmetric difference $(a \land \cmp{b}) \lor (b \land \cmp{a})$; note that $a = b$ \tiff $a \symdiff b = \bot$. If $\B$ is atomic, $FC(\B)$ is the finite--cofinite Boolean subalgebra of $\B$, i.e. every $b \in FC(\B) \setminus \set{\bot,\top}$ is a finite sum of atoms or the complement of such an element.

Suppose that $f$ is an n--ary operator on $B$.

\begin{enumerate}[label=\textup{(}\emph{\roman*}\textup{)}]
\item $f$ is called \emph{additive in its i--th argument}, if
\begin{gather*}
f(a_0, \ldots, a_{i-1},x, a_{i+1}, \ldots,a_{n-1}) \lor f(a_0, \ldots, a_{i-1},y, a_{i+1}, \ldots,a_{n-1}) = f(a_0, \ldots, a_{i-1},x \lor y, a_{i+1}, \ldots,a_{n-1}).
\end{gather*}
\item $f$ is called \emph{normal in its i--th argument} if $f(a_0, \ldots, a_{i-1},\bot, a_{i+1}, \ldots,a_{n-1}) = \bot$.
\end{enumerate}

Note that an additive operator is \emph{isotone}, i.e. it preserves the Boolean order in each of its arguments.

A \emph{Boolean algebra with operators} (BAO) is a Boolean algebra with additional mappings of finitary rank that are additive and normal in each argument \cite{jt51}.

A \emph{(unary) discriminator function} on $\B$ is an operation $d$ on $\B$ such that for all $a \in B$,
\begin{gather}\label{discr2}
d(a) =
\begin{cases}
\bot, &\text{if } a = \bot, \\
\top, &\text{otherwise}.
\end{cases}
\end{gather}
If $\B$ has a discriminator function, we call $\B$ a \emph{discriminator algebra}.

For a class $\K$ of BAOs, a unary term $t$ is a \emph{discriminator term} if it represents the discriminator function on each
subdirectly irreducible member of $\K$. A variety of BAOs is called a \emph{discriminator variety} if it is generated by a class of algebras
with a common discriminator term. 

Having a discriminator function $d$ allows us to convert satisfiability (validity) of inequations into satisfiability (validity) of equations:
Suppose that $\tau(\vec{x})$ and $\sigma(\vec{x})$ are terms with variables $\vec{x}$. Then
\begin{align}
(\exists \vec{x})[\tau(\vec{x}) \neq \sigma(\vec{x})] &\Iff (\exists \vec{x})[\tau(\vec{x}) \symdiff \sigma(\vec{x}) \neq \bot]
\Iff (\exists \vec{x})[d(\tau(\vec{x}) \symdiff \sigma(\vec{x})) = \top], \label{disceq0} \\
(\forall \vec{x})[\tau(\vec{x}) \neq \sigma(\vec{x})] &\Iff (\forall \vec{x})[\tau(\vec{x}) \symdiff \sigma(\vec{x}) \neq \bot] \Iff (\forall \vec{x})[d(\tau(\vec{x}) \symdiff \sigma(\vec{x})) = \top]. \label{disceq1}
\end{align}

If $\K$ is a class of algebras of the same type, we denote by $\K^d$ the class
obtained from adding a unary operation symbol which represents the discriminator function on the
members of $\K$.

\subsection{Complex algebras}
Traditionally, a subset of a group $G$ is called a \emph{complex of $G$}; the \emph{power algebra} of $G$ has $2^G$ as its universe, and the group operations lifted to $2^G$. Complex algebras are a generalization of this situation and special instances of BAOs. Suppose that $\klam{\A,\cO}$ is an algebra, and $f \in \cO$ is $n$--ary.  The \emph{complex operation} $\mathbf{f}: (2^A)^n \to 2^A$ corresponding to $f$ is defined by
\begin{gather}\label{cmf}
\mathbf{f}(a_0, \ldots, a_{n-1}) = \set{f(x_0, \ldots, x_{n-1}): x_0 \in a_0, \ldots, x_{n-1} \in a_{n-1}}.
\end{gather}
The \emph{full complex algebra of $\A$}, denoted by $\Cm\A$, has as its universe the powerset of $A$ and, besides the Boolean set operations, for each $f \in \cO$ its complex operator $\mathbf{f}$ defined by \eqref{cmf}.

More generally, the \emph{full complex algebra $\Cm\U$} of a relational structure  $\klam{U,\R}$ is the algebra $\klam{2^U, \cup,\cap, {}^-, \z,U}$, which has for every $R \in \R$ of, say, arity $n+1$, an $n$ -- ary operator
$f_R: (2^U)^n \to 2^U$ defined by
\begin{gather}\label{cm}
f_R(X_0, \ldots, X_{n-1}) = \set{y \in U: (\exists x_0,\ldots,x_{n-1})[x_0 \in X_0, \ldots, x_{n-1} \in X_{n-1} \tand R(y, x_0, \ldots, x_{n-1})]},
\end{gather}
see e.g. \cite{gol89}.

Each subalgebra of $\Cm \A$ is called a \emph{complex algebra of $\A$}. Of particular interest for us are the subalgebra of $\Cm \A$ generated by the constants, which we denote by $\Cm_0 \A$, and the subalgebra of $\Cm \A$ generated by the singletons $\set{a}$, where $a \in A$; we denote this algebra by $\Cm_1 \A$. Then, $\Cm_0 \A$ is the smallest subalgebra of $\A$ and $\Cm_1 \A$ is the subalgebra of $\Cm \A$ generated by the atoms. Clearly, $\Cm_0 \A \leq \Cm_1 \A$, but the converse need not be true; an example will be given below.

\subsection{Boolean monoids}

The complex algebras of the various structures which we will consider have one or more commutative Boolean monoids as a reduct:
A \emph{commutative Boolean monoid} (CBM) is an algebra $\A = \klam{A,\lor,\land,{}^-,\bot,\top, \circ, e}$ such that
\begin{align}
& \klam{A,\lor,\land,{}^-,\bot,\top} \text{ is a Boolean algebra}.\\
& \klam{A,\circ,e} \text{ is a commutative monoid.} \\
& \label{normal} x \circ \bot= \bot. \\
& \label{additive} x \circ (y \lor z) = (x\circ y) \lor (x \circ z).
\end{align}
In the sequel, we let $c(x) = x \circ \top$; it is well known that $c$ is an additive closure operator on CBMs \cite{jip92}. Furthermore \cite[see e.g.][]{reich_diss},
\begin{lemma}\label{lem:cong}
\begin{enumerate}[label=\textup{(}\emph{\roman*}\textup{)}]
\item The class CBM is congruence distributive.
\item $I$ is a congruence ideal -- i.e. the kernel of a congruence -- on a CBM $\A$ \tiff $I$ is a Boolean ideal and $x \in I$ implies $c(x) \in I$ for all $x \in A$.
\item The principal (Boolean) ideal generated by $c(x)$ is the smallest congruence ideal containing $x$.
\end{enumerate}
\end{lemma}
An element $x \in A$ is called a \emph{congruence element} if $c(x) = x$. By Lemma \ref{lem:cong}(3), each principal congruence ideal $I$ of $\A$ is of the form $I = \set{y: y \leq x}$ for some congruence element $x$. Note that a CBM is simple -- i.e. has only two congruences -- \tiff it satisfies
\begin{gather}\label{simple}
(\forall x)[x = \bot \lor c(x) = \top].
\end{gather}

\section{Complex algebras of $\N$}

Let $\N = \klam{\omega, 0, +, \cdot, 1}$ be the semiring of natural numbers, and $\CmN = \klam{2^\omega, \cap, \cup, {}^-, \z, \omega, \set{0}, \plus, \set{1}, \mal}$ be its full complex algebra, i.e.
\begin{align*}
a \plus b &= \set{n+m: n \in a, \ m \in b}, \\
a \mal b &= \set{n \cdot m: n \in a, \ m \in b}.
\end{align*}

A function $F: (2^\omega)^n \to 2^\omega$ is called \emph{circuit definable} if there is a term $\tau(v_0, \ldots, v_{n-1})$ in the language of $\CmN$ such that $F(s_0, \ldots, s_{n-1}) = \tau(s_0/v_0, \ldots, s_{n-1}/ v_{n-1})$ for all $s_0, \ldots, s_{n-1} \subseteq \omega$.
A subset $a$ of $\om$ is called \emph{circuit definable}, if there is a closed (i.e. variable free) term $\tau$ that evaluates to $a$. Each element of the smallest subalgebra $\CmNN$ of $\CmN$ corresponds to an arithmetic circuit with finite input nodes and vice versa via the interpretation $I$.

Both $\klam{2^\omega, \plus, \set{0}}$ and $\klam{2^\omega, \mal, \set{1}}$ are commutative monoids. Furthermore, $\plus$ and $\mal$ are normal and (completely) additive operators with respect to $\cup$, so that $\CmN$ is a Boolean algebra with operators, and
$$\klam{2^\omega, \cup, \cap, {}^-, \z, \omega, \plus, \set{0}}, \quad \klam{2^\omega, \cup, \cap, {}^-, \z, \omega, \mal, \set{1}}$$
are CBMs.

\begin{thm}\label{thm:disc0}
\begin{enumerate}[label=\textup{(}\emph{\roman*}\textup{)}]
\item \CmN\ is a discriminator algebra.
\item $\CmNN = \CmEN$.
\item $\CmNN$ is embeddable into any simple algebra of $\Var(\CmN)$.
\end{enumerate}
\end{thm}
\begin{proof}
(i) \ Let $f(x)$ be the function $\om \plus (\set{0} \mal x)$. If $x = \z$, then $\set{0} \mal x = \z$, and thus, $f(x) = \om \plus \z = \z$.
If $x \neq \z$, then $\set{0} \mal x = \set{0}$, hence $f(x) = \om \plus \set{0} = \omega$.

(ii) \ The atoms of $\CmN$ are the singletons $\set{n}$, and $\set{n} = \underbrace{\set{1} \plus \ldots \set{1}}_{n \text{ times}}$ if $n > 0$.

(iii) Since $\CmN$ is a discriminator algebra, it suffices to show that the smallest subalgebra $\A$ of an ultrapower of copies of $\CmN$ is isomorphic to $\CmNN$. Thus, let $\B \df {}^\kappa \CmN/U$ be an ultrapower of $\CmN$. Suppose that $e: \CmN \to \B$ is the canonical embedding, i.e. $e(a) = f_a/U$, where $f_a(i) = a$ for all $i < \kappa$. Since $\CmNN$ is generated by $\set{0}$, $e[\CmNN]$ is generated by $e(\set{0})$, and thus, since $e$ is an embedding, $e[\CmNN]$ is the smallest subalgebra of $\B$.
\end{proof}

\begin{thm}
The Boolean reduct of $\CmNN$ has $2^\omega$ ultrafilters.
\end{thm}
\begin{proof}
Let $p_0, \ldots, p_k, q_0, \ldots, q_k$ be different primes; then
\begin{gather*}
p_0 \cdot \ldots \cdot p_k \in (\omega \mal {p_0}) \land \ldots \land (\omega \mal {p_k}) \land (\cmp{\omega \mal \set{q_0}}) \land \ldots \land (\cmp{\omega \mal \set{q_k}}).
\end{gather*}
Hence, $\set{\omega \mal \set{p}: p \text{ prime}}$ is an independent set which generates an atomless Boolean subalgebra $\A$ of $\CmNN$. $\A$ has $2^\omega$ ultrafilters, and thus, so has $\CmNN$.
\end{proof}

The \emph{atom structure $\At\CmN$ of $\CmN$} has the set $\Omega = \set{\set{n}: n \in \omega}$ as its universe, and for each n -- ary operator $f$ an n+1--ary relation $R_f\df \set{\klam{p,q}: p \in \Omega^n \tand q \in \Omega, q \subseteq f(p)}$. Then,
\begin{align*}
R_{\plus}(\set{k},\set{n},\set{m}) &\Iff \set{m} \subseteq \set{k} \plus \set{n} \Iff k+n = m, \\
R_{\mal}(\set{k},\set{n},\set{m}) &\Iff \set{m} \subseteq \set{k} \mal \set{n} \Iff k \cdot n = m.
\end{align*}
It is well known that $\At\CmN \cong \N$. Let us call a relation on $\At\CmN$, i.e. on $\N$, \emph{circuit definable} if it corresponds to a circuit definable operator on $\CmN$.  A striking example of the lack of expressiveness of arithmetic circuits is the following:
\begin{thm}
\begin{enumerate}[label=\textup{(}\emph{\roman*}\textup{)}]
\item In $\N$, the converse $\geq$ of the natural ordering is circuit definable, while $\leq$ is not.
\item Relative subtraction is not circuit definable.
\end{enumerate}
\end{thm}
\begin{proof}
Using \eqref{cm} it is easily seen that $\geq$ is the relation corresponding to the function defined by $f(x) = x \plus \omega$. The ordering $\leq$ on $\omega$ corresponds to the function defined by $f(x) = \set{n \in \omega: (\exists m)[m \in x \tand n \leq m}$, and we have shown in \cite{pd_acfunc} that this function is not circuit definable. In the same paper we have proved (ii).
\end{proof}

\subsection{Complex algebras of $\klam{\omega, +, 0}$}

Let $\N^{\plus} = \klam{\omega,+,0,}$, $\CmN^{\plus}$ be its full complex algebra, and $\V$ be the variety generated by $\CmN^{\plus}$. Furthermore, set $c(x) = x \plus \omega$. Recall that the constant $\set{1}$ is definable in $\CmN^{\plus}$ by
\begin{gather*}
\set{1} = \omega \setminus ((\omega \setminus \set{0} \plus \omega \setminus \set{0}) \cup \set{0}).
\end{gather*}
Note that for all $a \subseteq \omega$,
\begin{align}
&c(a) = a \plus \omega = \set{k: (\exists n, m)[m \in a \tand k = m + n ]} = \set{k: \min(a) \leq k} = c(\set{\min(a)}), \\
&c(a) + \set{1} = a \plus c(\set{1}) = a \plus \cmp{\set{0}}.
\end{align}

The following observation will be useful:
\begin{lemma}\label{lem:fco}
Let $a,b \subseteq \omega$. Then, $a = \z$ or $b = \z$ \tiff $c(a) \cap c(b) = \z$.
\end{lemma}
\begin{proof}
If, say, $a = \z$, then $c(a) = \z$. Conversely, if $c(a) \cap c(b) = \z$, then one of $c(a)$ or $c(b)$ must be empty, since the intersection of any two cofinite sets is not empty. Hence, $a = \z$ or $b = \z$.
\end{proof}

Recall that $\CmNN^{\plus}$ is the smallest subalgebra of $\CmN^{\plus}$. The following result is well known:
\begin{lemma}\label{lem:fc}
The universe of $\CmNN^{\plus}$ is the finite -- cofinite subalgebra of $2^\om$.
\end{lemma}
Next, we describe the congruences of $\CmN^{\plus}$:

\begin{thm}\label{thm:Cn+}
The congruences of $\CmN^{\plus}$ form a chain of order type $1+ \omega^*$.
\end{thm}
\begin{proof}
By Lemma \ref{lem:cong}, $c(\set{n})$ is a congruence element generating the congruence $\theta_n$. Conversely, suppose that $\equiv$ is a congruence induced by the non--trivial ideal $I$; then, $I \neq \z$, and $I$ is closed under $c$. Since $I$ is also closed under $\subseteq$, $\set{\min(a)} \in I$ for every $a \in I$, and therefore, $n \df \min(\set{\min(a): a \in I, \ a \neq \z})$ exists, and $c(\set{n}) \in I$. If $a \in I, \ a \neq \z$, then $n \leq \min(a)$, and it follows that $a \subseteq c(a) = c(\set{\min(a)}) \subseteq c(\set{n})$. Hence, $I$ is the principal ideal of $2^\omega$ generated by $c(\set{n})$.

Observing that $c(\set{n}) = \set{m: n \leq m}$, we see that
\begin{gather*}
\z \subsetneq \ \ldots \subsetneq c(\set{n+1}) \subsetneq c(\set{n}) \subsetneq \ldots \subsetneq c(\set{1}) \subsetneq c(\set{0}) = \omega,
\end{gather*}
and thus,
\begin{gather}\label{cchain}
1' \subsetneq \ldots \subsetneq \theta_{n+1} \subsetneq \theta_n \subsetneq \ldots \subsetneq \theta_1 \subsetneq \theta_0 = V,
\end{gather}
where $1'$ is the identity and $V$ the universal congruence. Clearly, this chain has order
type $1 + \omega^*$. It follows that $\CmN^+$ has no smallest nontrivial congruence, and therefore,
$\CmN^+$ is not subdirectly irreducible.
\end{proof}

\begin{cor}\label{cor:fccong}
The congruences of $\CmNN^{\plus}$ form a chain of order type $1+ \omega^*$.
\end{cor}
\begin{proof}
Each congruence $\theta_n$ of $\CmN^{\plus}$ is generated by a cofinite congruence element,
which is in $\CmNN^{\plus}$ by Lemma \ref{lem:fc}.
\end{proof}

Let $\B_{n} \df \CmN^{\plus}/\theta_{n+1}$, and $\pi_n: \CmN^{\plus} \onto \B_n$ be the quotient mapping. Note that the kernel of $\theta_{n+1}$ is the ideal of $2^\omega$ generated by $c(\set{n+1}) = \set{n+1} \plus \omega = \omega \setminus [0,n]$. Thus, the Boolean part of $\B_n$ is isomorphic to the powerset algebra of $\set{0, \ldots, n}$ with atoms $g_i \df \pi_n(\set{i})$ for $i \leq n$. In particular, $\B_0$ is isomorphic to the two element Boolean algebra, since $c(\set{1}) = \omega \setminus \set{0}$ generates a prime ideal of $2^\omega$.

The composition table for $\circ$ on the atoms of $\B_n$ is given below. Observe that $g_0 = \pi_n(\set{0})$ is the identity element $e$ of $\klam{\B_n,\circ}$, and $g_m = g_0 \circ \underbrace{g_1 \circ \ldots\circ g_1}_{\text{m -- times}}$.

$$
\begin{array}{c|cccccc}
\circ & {g_0} & {g_1} & g_2 & \ldots & g_{n-1} & g_n \\ \hline
g_0 & g_0 & {g_1} & g_2 & \ldots & g_{n-1} & g_n \\
g_1 & g_1 & g_2 & g_3 &\ldots & g_n & \bot \\
g_2 & g_2 & g_3 & g_4 & \ldots & \bot & \bot \\
\ldots \\
g_n & g_n & \bot & \bot & \ldots & \bot & \bot
\end{array}
$$


\begin{thm}
\begin{enumerate}[label=\textup{(}\emph{\roman*}\textup{)}]
\item Each $\B_n$ is subdirectly irreducible.
\item $\Var(\B_n) \subsetneq \Var(\B_{n+1})$.
\item $\V = \Var \set{\B_n: n \in \om}$, and thus, $\V$ is generated by its finite members.
\end{enumerate}
\end{thm}
\begin{proof}
(i) \ The congruences of $\B_n$ are in 1--1 correspondence to the congruences of
$\CmN^{\plus}$ containing $\theta_n$. This is a finite chain, and the smallest nonzero congruence element of $\B_n$ is $g_n$.

(ii) \ Clearly, $\Var(\B_n) \subseteq \Var(\B_{n+1})$. In $\B_n$, $\underbrace{g_1 \plus \ldots \plus g_1}_{n+1 \text{ times}} = \bot$, and $\underbrace{g_1 \plus \ldots \plus g_1}_{n+1 \text{ times}} = g_{n+1} \neq \bot$ in $\B_{n+1}$.

(iii) \ Clearly, $\B_n \in \V$ for each $n \in \om$. Conversely, by Birkhoff's subdirect representation theorem \cite{bir35}, $\CmN^{\plus}$ is isomorphic to a subdirect product of its subdirectly irreducible quotients, see e.g. \cite{bs_ua}, Theorem 8.6. By Theorem \ref{thm:Cn+}, the only proper quotients of $\CmN^{\plus}$ are the algebras $\B_n$, and these are subdirectly irreducible by 1. above.
\end{proof}

$\V$ contains all Boolean algebras for which the extra operator $\circ$ is the Boolean meet and  $e = \top$, since the universe of $\B_0$ is the two element Boolean algebra, and $\B_0 \in \V$. Moreover,

\begin{thm}
$\Var(\B_n)$ is finitely based for each $n \in \om$. Hence, $\Eq\B_n$ is decidable for all $n \in \om$.
\end{thm}
\begin{proof}
Since $\Var(\B_n)$ is congruence distributive and $\B_n$ is finite, Baker's finite basis theorem \cite{baker77} implies that $\Var(\B_n)$ is finitely based for each $n \in \om$. The second claim follows from the fact that a finitely based variety which is generated by a finite algebra  has a decidable equational theory.
\end{proof}

\begin{cor}\label{cor:core}
$\Eq\V$ is co -- r.e.
\end{cor}
\begin{proof}
Given an equation $\tau = \sigma$ we can check whether $\tau = \sigma$ holds in $\B_0, \B_1, \ldots, $, since $\Eq B_n$ is decidable. Since $\V$ is generated by $\set{\B_n: n \in \om}$, any equation that fails in $\V$ must fail in some $\B_n$.
\end{proof}

Let $g$ be the term
\begin{gather}\label{def:g}
g \df \cmp{\cmp{e} \circ \cmp{e}} \land \cmp{e}.
\end{gather}
In $\CmN^{\plus}$, $g$ evaluates to $\set{1}$. Furthermore, we set
\begin{gather*}
g^n \df
\begin{cases}
e, &\text{if } n = 0, \\
\underbrace{g \circ g \circ \ldots \circ g}_{n-\text{times}}, &\text{otherwise.}
\end{cases}
\end{gather*}

Consider the following identities in the language of $\V$: 
\begin{align}
& \label{ge0} e \land (x \circ y) = e \land x \land y. \\
& \label{ge1} c(g^{n+1}) = \cmp{g^0 \lor \ldots \lor g^n} \text{ for all } n \in \omega. \\
& \label{ge2} c[c(x) \land \cmp{c(y)}] \land c[c(y) \land \cmp{c(x)}] = \bot.  \\
& \label{ge4} g \land (x \circ y) = [(e \land x) \circ (g \land y)] \lor [(g \land x) \circ (e \land y)] \\
& \label{ge5} (x \land g^n) \circ (\cmp{x} \land g^n) = \bot \text{ for all } n \in \omega. \\
& \label{ge7} c(x) = c(x \land \cmp{x \circ \cmp{e}}).
\end{align}
\begin{lemma}
\eqref{ge0} -- \eqref{ge7} hold in $\CmN^{\plus}$, and thus, in $\V$.
\end{lemma}
\begin{proof}
\eqref{ge0}: \ Just note that $0 \in a \plus b \Iff 0 \in a \tand 0 \in b$ so that $\set{0} \cap (a \plus b) \neq \z$ \tiff $\set{0} \cap a \neq \z$ and $\set{0} \cap b \neq \z$.

\eqref{ge1}: \ $c(\set{n+1}) = \set{n+1} \plus \omega = \uparrow\set{n+1} = \cmp{\set{0, \ldots,n}}$.

\eqref{ge2}: \ The set $\set{c(a): a \subseteq \omega}$ is a chain, thus, $c(x) \cap \cmp{c(y)} = \z$ or $c(y) \cap \cmp{c(x)} = \z$; hence,  $c(c(x) \cap \cmp{c(y)}) = \z$ or $c(c(y) \cap \cmp{c(x)}) = \z$. Now apply Lemma \ref{lem:fco}.

\eqref{ge4}: \ This follows immediately from the definition of $\plus$.

\eqref{ge5}: \ Each $g^n$ is an atom of $\CmNN^{\plus}$, so $x \land g^n = \bot$ or $\cmp{x} \land g^n = \bot$ for all $x \in \CmNN^{\plus}$.

\eqref{ge7}: \ If $a \subseteq \omega$ and $a = \z$, the claim clearly holds. If $a \neq \z$, then $a \cap \cmp{a \plus \cmp{\set{0}}} = \min(a)$ whence the conclusion follows.
\end{proof}

We do not know whether \eqref{ge0} -- \eqref{ge7} are sufficient to axiomatize $\V$.

\begin{thm}\label{thm:sip}
Let $\A \in \V$ be subdirectly irreducible and suppose that $d$ is the smallest nonzero congruence element in $\A$.
\begin{enumerate}[label=\textup{(}\emph{\roman*}\textup{)}]
\item $e$ is an atom of $A$.
\item The congruence elements of $\A$ are linearly ordered.
\item If \A\ is finite, then it is isomorphic to some $\B_n$.
\end{enumerate}
\end{thm}
\begin{proof}
(i) \  Assume that there are $a,b \in A$ such that $\bot \lneq a,b$, $a \land b = \bot$, and $a \lor b = e$. Then, the monotonicity of $\circ$ implies that $a \circ b \leq e \circ
e = e$, and by \eqref{ge0}, $a \circ b = (a \circ b) \land e = a \land b \land e = \bot$.

Since $a \neq \bot$, we have $a \circ \top \neq \bot$, and the fact that $d$ is the smallest non--zero congruence element implies $d \ \leq a \circ \top$. Now,
\begin{gather*}
d \leq a \circ \top \Implies d \circ b \leq a \circ b \circ \top = \bot \circ \top = \bot,
\end{gather*}
and, similarly, $d \circ a = \bot$. But then,
\begin{gather*}
d = d \circ e = d \circ (a \lor b) = (d \circ a) \lor (d \circ b) = \bot,
\end{gather*}
contradicting our hypothesis that $d \neq \bot$.

(ii) \ Assume there are nonzero congruence elements $x,y$ such that $x \land \cmp{y} \neq \bot$ and $y \land \cmp{x} \neq \bot$. Then, both $c(x \land \cmp{y})$ and $c(y \land \cmp{x})$ are nonzero congruence elements, and therefore, $d \leq c(x \land \cmp{y}) \land c(y \land \cmp{x})$. On the other hand,
\begin{gather*}
c(x \land \cmp{y}) \land c(y \land \cmp{x}) = c(c(x) \land \cmp{c(y)}) \land c(c(y) \land \cmp{c(x)}) = \bot
\end{gather*}
by \eqref{ge2} which contradicts $d \neq \bot$.

(iii) \ By \eqref{ge1}, $m \neq n$ implies that $g^m \land g^n = \bot$. Therefore, since $\A$ is finite, there exists a smallest $n$ such that $g^{n+1} = \bot$. We will prove that $\A = \B_n$.

\begin{enumerate}
\item $c(g^n) = g^n$: Again by \eqref{ge1} we have $g^0 \lor \ldots g^{n-1} \lor c(g^n) = \top$, and $c(g^n) \land g^m = \bot$ for all $m \lneq n$. Suppose there is some $s \in A$ such that $s \land g^n = \bot$ and $s \lor g^n = c(g^n)$. Then,
\begin{gather*}
g \circ s = \bot \lor (g \circ s) = (g \circ g^n) \lor (g \circ s) = g \circ (g^n \lor s) = g \circ c(g^n) = g \circ (g^n \circ \top) = (g \circ g^n) \circ \top = \bot,
\end{gather*}
and, by the normality of $\circ$ we obtain $g^m \circ s = \bot$ for all $1 \leq m \leq n$. Now,
\begin{gather*}
c(g^n) = g^n \circ (g^0 \lor \ldots \lor g^n \lor s) = (g^n \circ g^0) \lor \underbrace{g^n \circ g^1}_{ = \bot} \lor \ldots \lor
\underbrace{g^n \circ g^n}_{ = \bot} \lor \underbrace{g^n \circ s}_{ = \bot} = g^n.
\end{gather*}
It follows that $s = \bot$ and also that $g^0 \lor \ldots \lor g^n = \top$.
\item $d = g^n$: \ Since $d$ is the smallest congruence element, we have $d \leq g^n$. Assume there is some $t \neq \bot$ such that $d \land t = \bot $ and $d \lor t = g^n$. Then, for $x \in \set{d,t}$ and $y \in \set{g^1, \ldots g^n}$ we have $x \circ y = \bot$. Furthermore, $d \circ d = d \circ t = t \circ t = \bot$, since $d,t \leq g^n$ and $g^{n+1} = \bot$. This  implies that $d$ and $t$ are disjoint nonzero congruence elements, contradicting the subdirect irreducibility of $\A$. It follows that $d = g^n$.

\item Each $g^m$ is an atom of $\A$: Assume that there are $\bot \lneq s,t \lneq g^m$ with $s \land t = \bot$ and $s \lor t = g^m$ for some $m \leq n$. By (i) above, we have $1 \leq m$. From $s \leq g^m$ it follows that $s \circ g^k \leq g^m \circ g^k = g^{k+m} \neq g^ m$ for $k \neq 0$. Therefore,
\begin{gather*}
g^m \land (s \circ \top) = g^m \land (s \lor (s \circ g^1) \lor \ldots \lor (s \circ g^n)) = s.
\end{gather*}
Similarly we obtain $g^m \land (t \circ \top) = t$. Since $t$ and $s$ are nonzero and disjoint, $s \circ \top$ and $t \circ \top$ are incomparable congruence elements, contradicting (ii).
\end{enumerate}
\end{proof}

\begin{thm}
$\A\in \V$ is simple \tiff $\card{A} \leq 2$.
\end{thm}
\begin{proof}
Clearly, $\A$ is simple if it has at most two elements. Conversely, let $\A$ be simple. If $g \neq \bot$, then $c(g) = \top$ by \eqref{simple}, and thus, $\bot = \cmp{c(g)} = e$ by \eqref{ge1}.
The normality of $\circ$ implies that, for all $x \in A$, $x = e \circ x = \bot \circ x = \bot$, and therefore, $\A$ has only one element.

Now, suppose that $g = \bot$; then, $\bot = c(g) = \cmp{e}$ by \eqref{ge1}, and thus, $e = \top$. If $x \neq \bot$, then
\begin{gather*}
x = x \circ e = x \circ \top = c(x) = \top,
\end{gather*}
the latter by the simplicity of $\A$.
\end{proof}

Since every nontrivial variety contains a nontrivial simple algebra, it follows that the subvariety $\V_0$ of $\V$ generated by $\B_0$ is smallest nontrivial subvariety of $\V$.

If $\A$ is a CBM, we call $z \in A$ an \emph{annihilator of $\circ$}, if $x \circ z = z$ for all $x \in A, \ x \neq \bot$. The complex algebra of $\klam{\omega, 1, \cdot}$ has $\set{0}$ as a nonzero annihilator. This cannot happen in $\V$:

\begin{thm}\label{lem:anni}
Suppose that $\A \in \V$ and that $\card{A} > 2$. Then, $\A$ has no nonzero annihilator.
\end{thm}
\begin{proof}
Since $\V =\hom\sub\prods\set{\B_n: n \in \omega}$, there are a sequence $\set{\C_\alpha: \alpha < \kappa}$ of algebras from $\set{\B_n: n \in \omega}$, a subalgebra $\D$ of $\C \df \prod_{\alpha < \kappa} \C_\alpha$, and an onto homomorphism $\pi:\D \onto \A$ with kernel $I$. Let $g = \cmp{\cmp{e} \circ \cmp{e}} \land \cmp{e}$ in $\C$, and $g_\alpha = \cmp{\cmp{e} \circ \cmp{e}} \land \cmp{e}$ in $\C_\alpha$. Since $\D$ is a subalgebra of $\C$ and $g$ is a constant term, we have $g \in \D$; furthermore, $g(\alpha) = g_\alpha$ for all $\alpha < \kappa$.

Assume that $z$ is a nonzero annihilator of $\A$, and let $f \in \D$ with $z = \pi(f)$; since $z \neq \bot$ we have $f \not\in I$, in particular, $f \neq \bot$. Now, $z = z \circ \top = \pi(f) \circ \pi(\top) = \pi(f \circ \top)$, and we may suppose that $f$ is a congruence element. Since $\A$ has more than two elements, $\bot < g_\A$, and therefore $z \circ g_\A = z$. Hence, there is some $i \in I$ such that $(f \circ g) \lor i = f \lor i$, in particular, $f \leq (f \circ g) \lor i$; since $I$ is a congruence ideal, we may suppose \wlg that $i = c(i)$.

Let $\alpha < \kappa$ such that $f(\alpha) \neq \bot$, and suppose that $\C_\alpha = \B_n$; then, $f(\alpha) \leq (f(\alpha) \circ g_\alpha) \lor i(\alpha)$. Since $f$ is a congruence element, so is $f(\alpha)$, and it follows from the definition of $\B_n$ that there is some $m < n$ such that $f(\alpha) = c(g_\alpha^m)$. Now,
\begin{align*}
f(\alpha) &\leq (f(\alpha) \circ g(\alpha)) \lor i(\alpha) \\
&= (c(g_\alpha^m) \circ g_\alpha) \lor i(\alpha) \\
&=(g_\alpha^m \circ \top \circ g_\alpha) \lor i(\alpha) \\
&= c(g_\alpha^{m+1}) \lor i(\alpha) \\
&\overset{\eqref{ge1}}{=} \cmp{g_\alpha^0 \lor \ldots \lor g_\alpha^{m}} \lor i(\alpha) \\
&= (\cmp{g_\alpha^0 \lor \ldots \lor g_\alpha^{m-1}} \land \cmp{g_\alpha^m}) \lor i(\alpha) \\
&= (f(\alpha) \land \cmp{g_\alpha^m}) \lor i(\alpha), \\
\intertext{which implies}
f(\alpha) \land g_\alpha^m & \leq i(\alpha).
\end{align*}
Now, $f(\alpha) = \cmp{g_\alpha^0 \lor \ldots \lor g_\alpha^{m-1}}$ implies $g_\alpha^m \leq f(\alpha)$, and thus,
$g_\alpha^m \leq i(\alpha)$. Since $i(\alpha)$ is a congruence element, we have $i = i \circ \top$, and therefore,
\begin{gather*}
f(\alpha) = c(g_\alpha^m) = g_\alpha^m \circ \top \leq i(\alpha) \circ \top = i(\alpha).
\end{gather*}
Thus, $f(\alpha) \leq i(\alpha)$ for all $\alpha < \kappa$ and it follows that $f \in I$, contradicting our hypothesis.
\end{proof}

Let us briefly look at the complex algebra $\CmN^{+,\leq}$ of $\klam{\omega, +, \leq, 0}$. We have seen earlier that the complex version of
$\leq$ is the operator $\da : 2^{\omega} \to 2^{\omega}$ defined by $\da~a = \set{n \in \om: (\exists m)[m \in a \tand n \leq m]}$; thus, the universe of $\CmNN^{+,\leq}$ is $FC(\omega)$.

Since $\leq$ is first order definable in $\klam{\omega,+,0}$, one might suspect that $\CmN^{+,\leq}$ and $\CmN^{+}$ are not ``too far apart''. It turns out, however that $\CmN^{+,\leq}$ has much stronger properties than $\CmN^{+}$. 
\begin{thm}\label{thm:leq}
\begin{enumerate}[label=\textup{(}\emph{\roman*}\textup{)}]
\item $\Cm \N^{+,\leq}$ is a discriminator algebra.
\item $\Eq \Cm \N^{+,\leq} \neq \Eq \CmNN^{+,\leq}$.
\end{enumerate}
\end{thm}
\begin{proof}
(i) \ Set $d(x) \df \omega \plus \da x$. If $x = \z$, then $\da x = \z$, and thus, $d(x) = \z$. Otherwise, $0 \in \da x$, hence, $d(x) = \omega \plus x = \omega$.

(ii) \ Consider the function $\fin: 2^\om \to \set{\z,\om}$ defined by $\fin(a) \df \cmp{d(\cmp{\da a})}$. Then,
\begin{gather*}
\fin(a) =
\begin{cases}
\om, &\text{if } \card{a} = \om, \\
\z, & \text{if } a \text{ is finite.}
\end{cases}
\end{gather*}
Since for each $a \in \CmNN^{+,\leq}$, either $a$ finite or $\cmp{a}$ is finite, the equation $\fin(a) \cap \fin(\cmp{a}) = \z$ holds in $\CmNN^{+,\leq}$, but not in $\CmN^{+,\leq}$.
\end{proof}

\subsection{Complex algebras of $\klam{\omega, \cdot, 1}$}

Let $\N^{\mal} = \klam{\omega,\cdot,1}$, $\CmN^{\mal}$ be its complex algebra, and $\V$ be the variety generated by $\CmN^{\mal}$. Furthermore, let $c(a) \df \omega \mal a$ for every $a \subseteq \omega$.

We will first describe the smallest subalgebra of $\CmN^{\mal}$.


\begin{thm}\label{thm:cm0}
$\CmNN^{\plus} \cong \CmNN^{\mal}$.
\end{thm}
\begin{proof} For each $n \in \omega$, let
\begin{gather*}
a_n \df \set{m \in \omega: m \text{ has exactly $n$ (possibly repeated) prime divisors}}.
\end{gather*}
Then, $a_0 = \set{1}$, and the set of primes is circuit definable by
\begin{gather*}
a_1 = \cmp{(\omega \setminus \set{1}) \mal (\omega \setminus \set{1})} \setminus \set{1}
\end{gather*}
It comes as no surprise that $a_1$ is nothing else than the constant $g$ defined in \eqref{def:g}.
Each $a_n$ is circuit definable, since $a_n = \underbrace{a_1 \mal \ldots \mal a_1}_{n\text{--times}}$. Clearly, $a_i \cap a_j = \z$ for $i \neq j$, and $\bigcup_{n \in \om} a_n= \omega \setminus \set{0}$; the latter can be shown via induction on the degree of a term.

Let $A_0$ be the Boolean algebra with atoms $\set{\set{1}, \omega \setminus \set{1}}$, and for $n+1$ let $A_{n+1}$ be the Boolean closure of $\set{a \mal b: a,b \in A_n}$.Furthermore, for each $n \in \omega$, let
\begin{gather*}
b_{n+1} = \cmp{a_0 \cup \ldots \cup a_n}.
\end{gather*}
\begin{claim*}
For $0 < n$ each $A_n$ is finite with atoms $a_0, \ldots, a_{2^{n-1}}, b_{2^{n-1}+1}$.
\end{claim*}
First, we consider $n =1$. Computing $\set{a \mal b: a,b \in A_0}$), we retain $A_0$ (since $\set{1} \in A_0$) and, obtain additionally, $(\om \setminus \set{1}) \mal (\om \setminus \set{1})$ which is the set of all positive composite numbers. Thus, the atom $\omega \setminus \set{1}$ of $A_0$ splits into $a_1$, the set of all prime numbers, and $b_2$, the set of all composite numbers (including $0$). Since $1 = 2^{1-1}$, the claim is true for $n=1$.

Suppose that the claim is true for $A_n$, i.e. that the atoms of $A_n$ are $a_0, a_1, \ldots, a_{2^{n-1}}, b_{2^{n-1}+1}$. We need to show that the closure of $\set{a \mal b: a,b \in A_n}$ under the Boolean operations gives us $A_{n+1}$. Since $\mal$ distributes over $\cup$ it is sufficient to find $a_i \mal a_j$ and $a_i \mal b_{n+1}$ for $i,j \leq n$. Now, if $i,j \leq n$, then $a_i \mal a_j = a_{i + j}$, and thus, from $a_i \mal a_j$ we obtain the disjoint sets
\begin{gather*}
a_0, \ a_1, \ldots, a_{2^{n-1}}, \ a_{2^{n-1}+1}, \ldots, a_{2^{n-1} + 2^{n-1}} = a_{2^{(n+1)-1}}.
\end{gather*}
From $a_i \mal b_{2^{n-1}+1}$ we obtain
\begin{gather*}
b_{2^{n-1}+1} \supseteq b_{2^{n-1}+2} \supseteq b_{2^{n-1}+1 + 2^{n-1}} = b_{2^{(n+1)-1}+1}
\end{gather*}
The claim now follows from $b_{m} \setminus b_{m+1} = a_m$.

Clearly, $\set{a_n: n \in \omega}$ is the set of atoms of $\CmNN^{\mal}$. Let $f:\CmNN^{\plus} \to \CmNN^{\mal}$ be the mapping induced by $f(\set{n}) = a_n$. Then, $f$ is bijective, and
\begin{gather*}
f(\set{n} \plus \set{m}) = f(\set{n+m}) = a_{n+m} = a_n \mal a_m = f(\set{n}) \mal f(\set{m}).
\end{gather*}
Since $\plus$ and $\mal$ are (completely) additive, $f$ is an isomorphism.
\end{proof}
It may be noted that that $0 \in \cmp{a_n}$ for all $n \in \omega$. Thus, $\set{0}$ is not definable from the constants, and
\begin{gather*}
\omega = \sum\nolimits^{\CmNN^{\mal}} \set{a_n: n \in \omega} \neq \sum\nolimits^{\CmN^{\mal}} \set{a_n: n \in \omega} = \omega \setminus\set{0}.
\end{gather*}
It follows that $\CmNN^{\mal}$ as a Boolean algebra is not a regular Boolean subalgebra of $\CmN^{\mal}$ \cite[for the definition see][]{kop89}.

Let us now consider the algebra $\CmEN^\mal$, i.e. the subalgebra of $\CmN^\mal$ generated by its atoms $\set{n}$. We note that $\set{0}$ is a nonzero annihilator, and thus is a proper congruence element - indeed, the smallest nonzero congruence element. Therefore, $\CmEN^\mal$ is subdirectly irreducible. By Theorem \ref{lem:anni}, no element of $\Var\CmN^{\plus}$ with more than two elements has a nonzero annihilator. Together with $\CmNN^{\plus} \cong \CmNN^{\mal}$ we obtain that $\CmEN^\mal \not\in \Var \CmNN^\mal$, and therefore, $\Var\CmEN^\mal \neq \Var \CmNN^\mal$

Let $\theta$ be the congruence generated by $\set{0}$, and $\A$ be the complex
algebra of $\klam{\omega \setminus \set{0}, \cdot,1}$. Then, clearly, $a \theta b \Iff a \cup
\set{0} = b \cup \set{0}$, and $\CmEN^\mal/\theta$ is isomorphic to the
singleton algebra $\A_1$ of $\A$; furthermore, $\CmNN^\mal \cong \A_0$.

Owing to the presence of the nonzero annihilator $\set{0}$ we can still turn
satisfiability (validity) of inequations into satisfiability (validity) of
equations even though $\CmEN^\mal$ is not a discriminator algebra - it is subdirectly irreducible, but not simple:

\begin{align}
(\exists \vec{x})[\tau(\vec{x}) \neq \sigma(\vec{x})] &\Iff (\exists \vec{x})[\tau(\vec{x}) \symdiff \sigma(\vec{x}) \neq \bot]
\Iff (\exists \vec{x})[\set{0} \mal(\tau(\vec{x}) \symdiff \sigma(\vec{x})) = \set{0}], \label{discmal} \\
(\forall \vec{x})[\tau(\vec{x}) \neq \sigma(\vec{x})] &\Iff (\forall \vec{x})[\tau(\vec{x}) \symdiff \sigma(\vec{x}) \neq \bot] \Iff (\forall \vec{x})[\set{0}\mal (\tau(\vec{x}) \symdiff \sigma(\vec{x})) = \set{0}]. \label{discmal1}
\end{align}

We know already that the set of primes is definable in $\CmEN^\mal$. This can be generalized as follows: For $n \in \om$ let $\Po(n)$ be the set of all powers of $n$.
\begin{thm}\label{lem:plus} Let $p_0,\ldots, p_n$ be primes, and $b = \Po(p_0) \mal \ldots \mal
\Po(p_n)$. Then, for all $a \subseteq \omega$,
\begin{gather*}
 a \cap b \in \CmEN^\mal \Iff a \cap b \in FC(b).
\end{gather*}
Here, $FC(b)$ is the set of all finite or cofinite subsets of $b$.
\end{thm}
\begin{proof}
\al We first show that $\Po(p) \in \CmEN^\mal$ for every prime $p$.
Consider the following sequence:
\begin{xalignat*}{2}
&\om \mal \set{p} && \text{All multiples of $p$} \\
& \cmp{\om \mal \set{p}} && \text{All $n$ not divisible by $p$} \\
& \cmp{\om \mal \set{p}} \cap \cmp{\set{1}} && \text{All $n \neq 1$ not divisible by $p$, i.e. coprime to $p$, since $p$ is prime} \\
& \om \mal (\cmp{\om \mal \set{p}} \cap \cmp{\set{1}}) && \text{All $n$ with a factor $ \neq 1$ coprime to $p$} \\
& \cmp{\om \mal (\cmp{\om \mal \set{p}} \cap \cmp{\set{1}})} && \text{All $n$ with ($n \neq 1 \Implies$ no $m$ coprime to $p$ divides $n$),}
\end{xalignat*}
which defines the set of all powers of $p$. It follows that $b \in \CmEN^\mal$.
Since all singletons are in $\CmEN^\mal$, each finite or cofinite subset of $b$
is in $\CmEN^\mal$.

\ar Consider the condition
\begin{gather}\label{notFC}
 x \cap b \not \in FC(b).
\end{gather}

Suppose there are a term of minimal length $\tau(x_0, \ldots, x_k)$ and $a_0, \ldots, a_k \subseteq \omega$
such that $a \df \tau(a_0, \ldots, a_k)$ satisfies \eqref{notFC}. If $a = s \cup t$, then $s$
or $t$ satisfy \eqref{notFC}, contradicting the minimality of $\tau$; similarly, $a$ is
not of the form $\cmp{s}$. Finally, let $a = s \mal t$. By the minimality of $\tau$,
both $s \cap b$ and $t \cap b$ are in $FC(b)$, and by our assumption one must be cofinite
in $b$, say, $s$. The cofinality implies there are $q_0, \ldots, q_n \in \om$ such that
such that for all $m_0, \ldots, m_n \in \omega$,
\begin{gather}\label{malFC2}
 m_0 \geq q_0 \land \ldots \land m_n \geq q_n \Implies p^{m_0} \cdot \ldots \cdot p^{m_n} \in s.
\end{gather}
Let $q = p_0^{j_0} \cdot \ldots p_n^{j_n} \in t \cap b$. Then, $\set{q} \mal s$ is cofinite
in $s$ by \eqref{malFC2}, and thus, cofinite in $b$. It follows that $a = s \mal t$ is cofinite in $b$,
contradicting our assumption.
\end{proof}

Theorem \ref{lem:plus} does not hold in $\CmNN$: If $a \df (\set{3} \mal \om) \plus \set{1}$, then $\Po( 2) \cap a$ is the set of all powers of $4$. This also shows that $\CmEN^\mal \subsetneq \CmEN$.

\begin{thm}
\begin{enumerate}[label=\textup{(}\emph{\roman*}\textup{)}]
\item For each $n > 0$, $\CmEN^\mal$ contains an idempotent subsemigroup with $n$ generators and $2^n -1$ elements.
\item Suppose that $G$ is a subsemigroup of $\CmEN^\mal$ and a group. Then, $\card{G} = 1$.
\end{enumerate}
\end{thm}
\begin{proof}
(i) \ Let $P = \set{p_1,\ldots,p_n}$ be a set of $n$ primes, and for each nonempty $M = \set{p_{i_1}, \ldots p_{i_k}} \subseteq P$ let $a_M \df \Po p_{i_1} \mal \ldots \mal \Po p_{i_k}$. Then, $S = \set{a_M: \z \neq M\subseteq P}$ is the desired semigroup generated by $\set{a_{\set{p_i}}: 1 \leq i \leq n}$; the identity element is $a_P$.

(ii) \ Let $e$ be the neutral element of $G$. If $e = \z$, then $a = a \mal e = a \mal \z = \z$ for all $a \in G$, and thus, $\card{G} = 1$. Similarly, if $e = \set{0}$ we have $\card{G} = 1$. Thus, suppose that $e \not\subseteq \set{0}$; it is easy to see that then $a \not\subseteq \set{0}$ for all $a \in G$. Let $n = \min(e \setminus \set{0})$. Since $e \mal e = e$, there are $k,m \in e$ with $n = k \cdot m$. Minimality of $n$ and $n \neq 0$ imply $n = k$ and $m = 1$ or $n = m$ and $k = 1$. In any case, $n = 1$, and thus, $1 \in e$.

Suppose that $a \in G$. Since $a \mal \inv{a} = e$ and $1 \in e$, we have $1 \in a \cap \inv{a}$ and hence, $a = a \mal \set{1} \subseteq a \mal \inv{a} = e$.

Conversely, $e = e \mal \set{1} \subseteq e \mal a = a$, so that altogether $a = e$.
\end{proof}

\section{Decidability of theories}

Recall that for a BAO $\B$, we denote by $\B_0$ the smallest subalgebra of $\B$, i.e. the subalgebra of $\B$ generated by the constants. In this section we consider the problems $\FO\B$, $\Eq\B$, and $\EqSat\B$ for the algebras $\CmN$, $\CmNN$, $\CmN^{\plus}$, $\CmNN^{\plus}$, $\CmN^{\mal}$, and $\CmNN^{\mal}$. If $\B$ is one of these algebras, we denote by $\B^d$ the algebra enhanced by an additional operator $d$ which represents a discriminator function on $\B$. A \emph{conjunctive grammar} is a context--free grammar with an explicit intersection operation \cite{okh2001}. This section largely draws together work by \citet{Okhotin03}, \citet{cf:j+o08a}, and \citet{{PV05}}.

We have the following undecidability results. If T is a Turing Machine, we can define the language
$\mbox{VALC}(T)$ of {\em computations} of $T$, over the alphabet $\Sigma = \{0, \ldots, k-1\}$, for some $k>0$. It does not really matter how these computations are encoded: the important point here is that $\mbox{VALC}(T) = \emptyset$ if and only if the language accepted by $T$ is empty. We may assume without loss of generality that no strings in $\mbox{VALC}(T)$ begin with the letter $0$. 
Any string $s \in \Sigma^*$ which does not begin with $0$ may be regarded as a base-$k$ representation of a positive integer
$\sharp(s)$. Thus, we obtain a 1--1 mapping $f_k: \mbox{VALC}(T) \rightarrow \{a\}^*$ given by $f_k(s) = a^{\sharp(s)}$. Thus, $f_k(\mbox{VALC}(T))$ is a language over the 1-element alphabet $\{a\}$.

\begin{lemma} \label{lem:rec}
\cite{cf:j+o08a})
\begin{enumerate} [label=\textup{(}\emph{\roman*}\textup{)}]
\item For every Turing Machine $T$, we can effectively construct conjunctive grammars $G$ and $G'$ over the alphabet $\{a\}$ such that
$L(G) = f_k(\mbox{VALC}(T))$.
\item If $a \subseteq \omega$ is recursive, there exists a finite system of equations of the form $\tau_i(y,x_1,\ldots,x_n) = \sigma_i(y,x_1,\ldots,x_n)$ in the language with $\cup, \cap, \plus$ such that its unique solution is $y = a$ and $x_i = b_i$ for some $\klam{b_1,\ldots,b_n} \in (2^\omega)^n$.
\end{enumerate}
\end{lemma}

First, we compare the theories of these algebras.

\begin{thm}\label{thm:eq}
\begin{enumerate}[label=\textup{(}\emph{\roman*}\textup{)}]
\item $\Eq(\N^{\plus}) = \Eq(2^\omega, \plus, \set{0})$.
\item $\Eq\CmNN^{\plus} = \Eq\CmN^{\plus}$.
\item $\EqSat\CmNN^{\plus} \neq \EqSat\CmN^{\plus}$.
\item $\Eq\CmNN^{\plus,d}\neq \Eq\CmN^{+,d}$.
\item $\Eq\CmNN \neq \Eq\CmN$.
\end{enumerate}
\end{thm}
\begin{proof}
(i) \ The mapping $f: \omega \to \set{a \subseteq \om: a \text{ is finite}}$ which maps $n$ to $\set{n}$ is an embedding of monoids, and thus, $\Eq(2^\omega, \plus, \set{0}) \subseteq \Eq(\N^+)$. The reverse inclusion follows from the fact that $\N^+$ is the free monoid on a single generator.

(ii) \ Since $\CmNN^{\plus} \leq \CmN^{\plus}$, it follows that $\CmNN^{\plus} \in \Var(\CmN^{\plus})$. Conversely, each $\B_n$ is in $\Var(\CmNN^{\plus})$ by Corollary \ref{cor:fccong}, and thus, $\CmN^{\plus} \in \Var(\CmNN^{\plus})$.

(iii) \ The equation
\begin{gather}\label{e1}
x \plus \set{1} = \cmp{x}
\end{gather}
has a unique solution in $\CmN^{\plus}$, namely, the set of even numbers, which is not in $\CmNN^{\plus}$.

(iv) This is a slight generalization of Theorem \ref{thm:leq}(2). The equation \eqref{e1} has no solution in $FC(\om)$, i.e.
$(\forall x)[x \plus \set{1} \neq \cmp{x}]$ holds in $\CmNN$. This is equivalent to the equation $d((x \plus \set{1}) \symdiff \cmp{x}) = \omega$ which is not valid in $\CmN^{+,d}$.

(v) \ Let $a \in \CmN \setminus \CmNN$ be recursive. Such  set exists, since every every set definable by an arithmetic circuit is in the bounded hierarchy $\BH$ \cite{pd_acfunc}, and the bounded hierarchy is known to be contained within the zeroth
Grzegorczyk class, $\cE^0_*$. By Lemma \ref{lem:rec} there is a first order sentence $(\exists x)\phi(x)$ such that $\CmN \models \phi(x/s)$ \tiff $s = a$. It follows that $\CmNN \not\models (\exists x)\phi(x)$, i.e. $\CmNN \models (\forall x)\neg\phi(x)$. Since $\CmNN$ is a discriminator algebra, there is an equation $\tau(x) = \sigma(x)$, such that $\CmNN \models (\forall x)\neg\phi(x)$ \tiff $\CmNN \models \tau(x) = \sigma(x)$. Since $\CmN \models (\exists x)\phi(x)$, $\tau(x) = \sigma(x)$ cannot hold in $\CmN$.
\end{proof}

Given any conjunctive grammar $G$  with non-terminals $X_1, \ldots, X_n$ over the alphabet $\{a\}$, we may effectively construct a system of language equations $\Eqs$ in variables $V_1, \ldots, V_n$, with the property that $\Eqs$ has a unique least (under componentwise-inclusion) solution $S_1^0, \ldots, S_n^0$ and, moreover, for all $i$ ($1 \leq i \leq n$), $S_i$ is the set of strings of $\{a\}^*$ to which $G$ assigns the category $X_i$. Let us assume that $X_1$ is the start-symbol of $G$; i.e.,~$L(G)$ is the set of strings
to which $G$ assigns category $X_1$.
\begin{thm}
Let $O$ be any collection of isotone operators on $\N$ with $\plus \in O$. Then $\EqSat\Cm (\N,O)$ is co--r.e.-complete.
\label{theo:eqSatPlusCoReComplete}
\end{thm}
\begin{proof}
For the lower bound, it suffices to
establish the result in the case $\Cm (\N, \{+\}) = \Cm \N^{\plus}$.
We use the fact that the emptiness of the languages accepted by a
Turing machine $T$ is equivalent to the validity of the language
equations $\Eqs$, as outlined above. We must translate the language
equations in $\Eqs$ in the logical signature $\{\varepsilon, \{a\},
\cup, \cap, \cdot\}$, (where $\cdot$ denotes concatenation) into
integer-set equations, by replacing $\varepsilon$ by $\{0\}$, $\{a\}$
by $\{1\}$, and $\cdot$ by $\plus$. Let the result of this translation
be $\Eqs^*$. If $g: \{a\}^* \rightarrow \N$ is the isomorphism given
by $a^k \mapsto k$, then $S_1, \ldots, S_n$ is a solution of $\Eqs$ if
and only if $g(S_1), \ldots, g(S_n)$ is a solution of $\Eqs^*$.

Altogether, we have:
\begin{eqnarray*}
\mbox{Acc}(T) = \emptyset & \Leftrightarrow & \mbox{VALC}(T) = \emptyset \\
\ & \Leftrightarrow & f_k(\mbox{VALC}(T)) = \emptyset\\
\ & \Leftrightarrow & L_G = \emptyset\\
\ & \Leftrightarrow & S_1^0 = \emptyset\\
\ & \Leftrightarrow &
            \text{$\Eqs \cup \{X_1 = \emptyset\}$ has a solution}\\
\ & \Leftrightarrow &
            \text{$\Eqs^* \cup \{X_1 = \emptyset\}$ has a solution.}
\end{eqnarray*}
This establishes that $\EqSat\Cm (\N,O)$ is co-r.e.-hard, as
required.

To show that $\EqSat \Cm (\N,O)$ is co-r.e., it suffices to prove
that, for any $m$-tuple of variables $\bar{x}$ and any term
$\tau(\bar{x})$,
\begin{equation}
\tau(\bar{x}) = \emptyset \text{ has a solution in $(2^\omega)^m$}
\label{eq:solutionLimit}
\end{equation}
if and only if, for all $n$,
\begin{equation}
\tau(\bar{x}) \cap [0,n] = \emptyset \text{ has a solution in $(2^{[0,n]})^m$},
\label{eq:solutionStage}
\end{equation}
since the condition~\eqref{eq:solutionStage} is evidently decidable
for fixed $n$.  In the sequel, if $\bar{s} = (s_1, \ldots, s_m)$ and
$\bar{t} = (t_1, \ldots, t_m)$ are $m$-tuples of sets, we write
$\bar{s} \cap [0,n]$ for the $m$-tuple $(s_1 \cap [0,n], \ldots, s_n
\cap [0,n])$, $\bar{s} \cup \bar{t}$ for the
$m$-tuple $(s_1 \cup \bar{t}_1, \ldots, s_n
\cup \bar{t}_m)$
 and $\bar{s} \subseteq \bar{t}$ for the condition $s_1
\subseteq t_1 \wedge \cdots \wedge s_m \subseteq t_m$.

The direction from~\eqref{eq:solutionLimit}
to~\eqref{eq:solutionStage} is easy. For suppose $\tau(\bar{s}) =
\emptyset$. Then, for all $n$, $\tau(\bar{s}) \cap [0,n]= \emptyset$,
whence, by the monotonicity of the operators on $O$, $\tau(\bar{s}
\cap [0,n]) \cap = \emptyset$.  To show the converse, let $V_n$
denote, for any $n$, the set of pairs $\langle \bar{s},n \rangle$
where $\bar{s}$ is a solution of $\tau(\bar{x}) \cap [0,n] =
\emptyset$ in $(2^{[0,n]})^m$.  Thus, $V_n$ is finite, and,
assuming~\eqref{eq:solutionStage} for all $n$, non-empty.  Define the
directed graph $(V,E)$ by setting $V =
\bigcup V_n$ and
$$
E = \left\{\left(\langle \bar{s}, n \rangle, \langle \bar{t}, n+1 \rangle \right):
   \langle \bar{s},n \rangle \in V_n,\ \ \langle \bar{t}, n+1 \rangle \in V_{n+1}
   \text{ and } \bar{s} \subseteq \bar{t} \right\}.
$$
Thus, $(V,E)$ is a finitely branching, infinite tree, and so has an
infinite path $\langle \bar{s}_0, 0 \rangle, \langle \bar{s}_1, 1
\rangle, \ldots$, where $\bar{s}_0 \subseteq \bar{s}_1 \subseteq
\cdots$. Letting $\bar{s} = \bigcup \bar{s}_n$, we have, for all $n$,
$\tau(\bar{s}) \cap [0,n] = \tau(\bar{s}_n) \cap [0,n] =
\emptyset$. Hence $\tau(\bar{s}) = \emptyset$,
whence~\eqref{eq:solutionLimit} holds.
\end{proof}
It immediately follows from
Theorem~\ref{theo:eqSatPlusCoReComplete} that
\begin{cor}
$\EqSat \CmN$ is co-r.e.-hard.
\label{theo:eqSatPlusTimesCoReHard}
\end{cor}

In Corollary \ref{cor:core} we showed that $\Eq\CmNN^{\plus}$ is co--re. On the other hand, it is not obvious that we can find a (computable) bound for the smallest witnesses of {\em in}equations in these languages.

While the membership problem for $\CmNN$ is a word problem, the satisfaction problem \eqref{sc} is related to the equational theory: 

\begin{thm}
The equational theory of $\CmEN^\mal$ is decidable \tiff the satisfaction problem \eqref{sc} is decidable.
\end{thm}
\begin{proof}
\ar Let $n \in \omega$ and $\tau(\vec{x})$ be a term with variables $\vec{x}$. Then,
\begin{align*}
&\exists(\vec{x})[\set{n} \cap \tau(\vec{x}) \neq \z] \Iff \neg((\forall \vec{x})[\set{n} \cap \tau(\vec{x}) = \z]).
\end{align*}
\al Suppose that $\tau(\vec{x}), \ \sigma(\vec{x})$ are terms with variables among $\vec{x}$; w.l.o.g. we may suppose that $\sigma(\vec {x}) = \z$. Then,
\begin{align*}
&(\forall \vec{x})[\tau(\vec{x}) = \z] 
\Iff (\forall \vec{x})[0 \not \in \set{0} \mal \tau(\vec{x})]
\Iff \neg( (\exists \vec{x})[0 \in \set{0} \mal \tau(\vec{x})]).
\end{align*}
\end{proof}
As for equational theories, results are known as long as we have the
wherewithal to convert equations into inequations. Determining whether
an equation belongs to the equational theory of a language $\cL$ over
some interpretation $\fA$ is the co-problem of determining whether an
{\em in}equation in $\cL$ is satisfiable in $\fA$.  If we have a
discriminator at our disposal, then \eqref{disceq0}, \eqref{disceq1}
and Theorem~\ref{theo:eqSatPlusCoReComplete} imply
\begin{thm}
The set $\Eq\CmN^{\plus,d}$ is r.e.-hard. Hence, $\Eq\CmN$ is r.e.-hard.
\label{theo:eqPlusTimesReHard}
\end{thm}

If $\klam{S, \circ}$ is a semigroup, then its \emph{power structure} is the semigroup of complexes of $S$. The following result is quoted by \citet[Theorem 2.3.2]{PV05}:

\begin{thm}\label{thm:etsg}\cite{bay88} \
For a variety $\V$ of semigroups the class of power structures of elements of $\V$ has a decidable elementary theory if and only if $\V \subseteq \Var(\{x\circ y \circ z = x \circ z\})$.
\end{thm}

Neither $\klam{\omega, +, 0}$ nor $\klam{\omega, \cdot, 1}$ satisfy $x\circ y \circ z = x \circ z$. Since the power structure of $\klam{\omega, +, 0}$ is a reduct of $\CmN^{\plus}$, this is another way to show that $\FO\CmN^{\plus}$ is undecidable. It also applies to $\klam{\omega, \cdot, 1}$:
\begin{cor}
$\FO\CmN^{\mal}$ is undecidable.
\end{cor}


\bibliographystyle{apalike}
\bibliography{acalg}
\end{document}